%% file: main.tex
\documentclass[11pt]{article}
\usepackage[hmargin=1in,vmargin=1in]{geometry}
\usepackage{hchang}

\input{preamble}

\begin{document}

\title{
Max Cut with Small-Dimensional SDP Solutions\footnote{
GenAI tools were used in this paper. Our key technical lemma, Theorem~\ref{thm:geom}, was first proved by Google's internal Gemini model with a weaker bound $\eps n^2$ for unspecified $\eps(d) > 0$, which was reported in~\cite{woodruff2026accelerating}. 
Later, we used our personal subscriptions of ChatGPT 5.2 (``Extended Pro'') and Gemini 3.0 (``Pro DeepThink'') to obtain the optimal bound of $\eps = 2^{-\Theta(d)}$. Its proofs are edited by the authors and appear in Section~\ref{sec:geometry}.
}
}

\author{ 
{Hsien-Chih Chang} \\ Dartmouth University\\ \texttt{hsien-chih.chang@dartmouth.edu}
\and
{Suprovat Ghoshal} \\ Indiana University\\ \texttt{supghos@iu.edu}
\and
{Euiwoong Lee} \\ University of Michigan \\
\texttt{euiwoong@umich.edu}
}

\maketitle

\begin{abstract}
We study the Max-Cut semidefinite programming (SDP) relaxation in the regime where a
near-optimal solution admits a low-dimensional realization.
While the Goemans--Williamson hyperplane rounding achieves the worst-case optimal
approximation ratio \(\alpha_{\mathrm{GW}}\approx 0.87856\), it is natural to ask whether
one can beat \(\alpha_{\mathrm{GW}}\) when the SDP solution lives in \(\mathbb{R}^d\) for a
small dimension \(d\).
We answer this in the affirmative for every fixed \(d\):
there is a polynomial-time rounding algorithm that, given a \(d\)-dimensional feasible
solution to the standard Max-Cut SDP strengthened with triangle inequalities,
produces a cut of expected value at least
\((\alpha_{\mathrm{GW}}+2^{-O(d)})\) times the SDP value.
Our improvement is driven by a new geometric
anti-concentration lemma for signs of low-dimensional Gaussian projections.
\end{abstract}



\thispagestyle{empty}
\newpage

\section{Introduction}
\input{intro.tex}

\section{Key Geometric Theorem}
\label{sec:geometry}
\input{geometry.tex}

\section{Improved Algorithm}
\input{algorithms.tex}

\bibliographystyle{alpha}
\bibliography{mybib.bib}


\end{document}

%% file: preamble.tex
\usepackage{fullpage}

\usepackage{graphicx}
\graphicspath{{figs/}}
\usepackage{amsmath,amssymb,amsfonts}
\usepackage{xspace}
\usepackage{thm-restate}
\usepackage{mathtools}

\usepackage{algorithm}
\usepackage{algpseudocode}
\algtext*{EndWhile}
\algtext*{EndIf}
\algtext*{EndFor}
\algrenewcommand\algorithmicrequire{\textbf{Input:}}
\algrenewcommand\algorithmicensure{\textbf{Output:}}

\usepackage{enumitem}
\setlist{topsep=3pt, itemsep=0pt}

\newtheorem{theorem}{Theorem}[section]
\newtheorem{lemma}[theorem]{Lemma}
\newtheorem{claim}[theorem]{Claim}

\newtheorem{proposition}[theorem]{Proposition}

\newtheorem{corollary}[theorem]{Corollary}

\newcommand{\eps}{\varepsilon}

\allowdisplaybreaks

%% file: intro.tex

Max-Cut is a canonical combinatorial optimization problem of partitioning the vertex set of a
weighted undirected graph \(G=(V,E,w)\) into two sides so as to maximize the total weight of
edges crossing the cut. One of Karp's 21 NP-complete problems~\cite{karp2009reducibility}, Max-Cut has served as a testing ground for
new ideas in approximation algorithms and hardness of approximation.


From the algorithms side, the Goemans-Williamson (GW) algorithm~\cite{GW95} is one of the most well-known SDP-based approximation algorithms. It solves the basic semidefinite relaxation
\begin{equation}
  \max \ \ \frac12 \sum_{(i,j)\in E} w_{ij}\bigl(1-\langle v_i,v_j\rangle\bigr)
  \qquad \text{s.t. } \|v_i\|_2=1 \text{ for all } i\in V,
\label{eq:sdp-basic-intro}
\end{equation}
and rounds a feasible solution by sampling a random hyperplane through the origin and cutting
according to the induced signs:
sample \(g\sim\mathcal{N}(0,I)\) and output \(x_i=\mathrm{sgn}(\langle g,v_i\rangle)\).
For an edge \((i,j)\), the cut probability is \(\Pr[x_i\neq x_j]=\arccos(\langle v_i,v_j\rangle)/\pi\),
and the worst-case ratio between this probability and the SDP contribution \((1-\langle v_i,v_j\rangle)/2\)
is the famous Goemans-Williamson constant
\[
  \alpha_{\mathrm{GW}} \ :=\ \min_{\rho\in[-1,1]}
  \frac{\arccos(\rho)/\pi}{(1-\rho)/2}\ \approx\ 0.87856,
\]
which shows that this algorithms is an $\alpha_{GW}$-approximation.
Moreover, assuming the Unique Games Conjecture~\cite{Khot02}, no polynomial-time algorithm can achieve a
\((\alpha_{\mathrm{GW}}+\varepsilon)\)-approximation for any constant \(\varepsilon>0\)
\cite{KKMO07}. Thus, improving upon \(\alpha_{\mathrm{GW}}\) in the worst case is believed impossible.

Under what conditions could we improve the $\alpha_{GW}$-approximation by a constant? There have been active studies on this question, mainly relying on structural assumptions on the graph $G$, including when $G$ is {\em dense} (i.e., $|E| = \Omega(n)$)~\cite{arora1995polynomial} or {\em low-degree} (i.e., the maximum degree is $O(1)$)~\cite{FKL02, HK23}. 

\paragraph{Low-dimensional SDP solutions.}
In this paper, we study the
\emph{low-dimensionality} assumption: the optimal (or near-optimal) SDP solution happens to admit a realization in
\(\mathbb{R}^d\) for small \(d\).
It departs from previous results in that we pose structural assumptions on the SDP solution rather than the graph. 

Avidor and Zwick~\cite{AZ05} initiated the systematic study of this topic.
They showed that in two dimensions, one can achieve the ratio
\(\frac{32}{25+5\sqrt5}\approx 0.884458\) (originally obtained by Goemans in unpublished work),
and in three dimensions they designed the ``yin--yang'' rounding procedure with ratio about \(0.8818\).
This resolves in the negative a question of Feige and Schechtman~\cite{FS01} asking whether one can
have integrality ratios arbitrarily close to \(\alpha_{\mathrm{GW}}\) while admitting a three-dimensional
optimal embedding. 

Avidor and Zwick conjectured that for every fixed \(d\) there is a
rounding algorithm beating \(\alpha_{\mathrm{GW}}\) on instances whose optimal SDP solution lies in
\(\mathbb{R}^d\).
The conjecture is stated with respect to the basic SDP~\eqref{eq:sdp-basic-intro}, but they also asked whether the same conclusion can be made with respect to the SDP with the following {\em triangle inequality constraints}.
\[
  \|a_i v_i-a_j v_j\|_2^2 + \|a_j v_j-a_k v_k\|_2^2 \ \ge\ \|a_i v_i-a_k v_k\|_2^2
  \quad
  \mbox{ for every } i,j,k\in V \mbox{ and }a_i,a_j,a_k\in\{\pm1\},
\]

Prior to this work, no such improvement was known beyond \(d=2,3\) with or without the triangle inequalities. 

The lack of progress on this conjecture is somewhat surprising, given that the positive answer has been known for the related {\em Grothendieck} problems. 
In the (bipartite) Grothendieck problem, given $A \in \mathbb{R}^{m \times n}$, the goal is to compute $x \in \{ \pm 1 \}^m$ and $y \in \{ \pm 1 \}^n$ to maximize $x^T Ay$. In the PSD Grothendieck problem, given positive semidefinite $A \in \mathbb{R}^{n \times n}$, the goal is to compute $x \in \{ \pm 1 \}^n$ to maximize $x^T Ax$. The optimal approximation ratios for these problems, assuming the UGC, are {\em the Grothendieck constant} $K_G$ and $\pi/2$ respectively~\cite{raghavendra2009towards, briet2015tight}, but for both problems, bounded-dimension SDP solutions can be rounded with a strictly better approximation ratio~\cite{briet2010positive, briet2014grothendieck}.

\paragraph{Our result.}
We give a positive answer to the question of Avidor and Zwick with the triangle inequalities, albeit with an exponentially small improvement.

\begin{theorem}[Main Theorem]\label{thm:main}
There exists a randomized polynomial-time algorithm with the following guarantee.
Given a feasible \(d\)-dimensional solution \(\{v_i\in\mathbb{R}^d : \|v_i\|_2=1\}_{i\in V}\) to the
Max-Cut SDP strengthened with triangle inequalities, the algorithm outputs a cut
\(x\in\{\pm1\}^V\) whose expected value is at least
\[
  \bigl(\alpha_{\mathrm{GW}} + 2^{-O(d)}\bigr)\cdot \mathrm{SDP},
\]
where \(\mathrm{SDP}\) is the objective value of the input solution.
\end{theorem}

The algorithm is a local-improvement post-processing of GW rounding, following the framework of
Feige--Karpinski--Langberg~\cite{FKL02} and its sharpened analysis by Hsieh--Kothari~\cite{HK23} for Max-Cut with the bounded (maximum) degree.
Our contribution is a new geometric theorem for low-dimensional Gaussian signs, which lets us replace
the dependence on the (possibly unbounded) graph degree by an explicit dependence on the SDP
dimension.

Finally, we also mention that our proof directly extends to the weaker {\em local dimension} condition on the SDP solution $\{ v_i \}_{i \in V}$: for every $i \in V$, the vectors of its neighbors (i.e., $\{ v_j : (j, i) \in E \}$) lie in a $d$-dimensional subspace. In particular, our result holds for any graphs with maximum degree $d$ as well, though our improvement of $2^{-O(d)}$ over $\alpha_{GW}$ is much worse than the $1/\widetilde{\Omega}(d^2)$ via the direct analysis~\cite{HK23}. 

\subsection{Technical overview}\label{subsec:overview}
The proof of Theorem~\ref{thm:main} decomposes into a geometric statement and its algorithmic
application.

\paragraph{1. A geometric anti-concentration theorem.}
Let \(v_1,\dots,v_n\) be unit vectors in \(\mathbb{R}^d\) such that
\(\langle v_i,v_j\rangle \ge -0.9\) for all \(i\neq j\).
Sample a random Gaussian vector \(g\sim \mathcal{N}(0,I_d)\) and define
\(X_i=\mathrm{sgn}(\langle g,v_i\rangle)\in\{\pm1\}\) and \(X=\sum_{i=1}^n X_i\).
Our main geometric theorem (proved in Section~\ref{sec:geometry}) shows that
\[
  \E[X^2] \ \ge\ n^2\cdot 2^{-O(d)}.
\]
In other words, the random hyperplane induced by \(g\) produces a nontrivial imbalance on the vectors
\(\{v_i\}\).
We present three self-contained proofs, each emphasizing a different viewpoint (tensor moments,
Gaussian/Hermite analysis, and positive-definite kernels on the sphere).

\paragraph{2. From geometry to an improved Max-Cut rounding.}
Given the geometric theorem, our algorithm follows the framework of
Feige--Karpinski--Langberg~\cite{FKL02}. After performing the standard hyperplane rounding algorithm (i.e., sampling a random Gaussian $g \sim \mathcal{N}(0, I)$ and setting $x_i \leftarrow \sgn(\langle g, v_i \rangle)$), the algorithm tries to improve the solution further by performing the local improvement. 

More precisely, let $C := \{ i \in V: |\langle g, v_i \rangle| \leq \eps \}$ for some $\eps > 0$ be the set of {\em candidates}. We flip a candidate $i$ to the other value (i.e., $x_i \leftarrow -x_i$) if the {\em local gain}, defined as 
\begin{equation*}
\sum_{(i, j) \in E: j \notin C \mbox{ and } x_i = x_j} w(i, j) -
\sum_{(i, j) \in E: j \in C \mbox{ or } x_i \neq x_j} w(i, j),
\end{equation*}
is strictly positive; it words, we execute the flip of $i$ if its gain from its non-candidate neighbors is strictly greater than its loss from its non-candidate neighbors {\em plus} its total weight to all the candidate neighbors. 
This conservative flip strategy ensures that (1) the order of the flips is irrelevant, (2) the Max-Cut value never decreases by any flip, and (3) it suffices to analyze the expected {\em local gain} at each $i$ conditioned on $i$ being a candidate.

The Hsieh--Kothari analysis~\cite{HK23} shows that the expected {\em local gain} can be expressed in terms of
correlations among the signs of certain Gaussian projections of the neighbors' SDP vectors, and 
proves a quantitative lower bound on this correlation by a factor polynomial in the degree.

In our setting, the SDP vectors already live in \(\mathbb{R}^d\), hence the relevant Gaussian process has
rank at most \(d\) regardless of the graph degree.
The geometric theorem above provides a correlation lower bound that only depends on the dimension $d$,
namely \(2^{-O(d)}\), which is enough to turn the local improvement into an additive gain of
\(2^{-O(d)}\) times the total edge weight.


%% file: geometry.tex
\newcommand{\Pnl}{\Delta}
\newcommand{\pnl}{\delta}

\newmuskip\pFqmuskip
\newcommand{\pFqcomma}{{\normalcomma}\mskip\pFqmuskip}
\newcommand*\pFq[6][8]{%
  \begingroup 
  \pFqmuskip=#1mu\relax
  \mathchardef\normalcomma=\mathcode`,
  \mathcode`\,=\string"8000
  \begingroup\lccode`\~=`\,
  \lowercase{\endgroup\let~}\pFqcomma
  {}_{#2}F_{#3}{\left[\genfrac..{0pt}{}{#4}{#5};#6\right]}%
  \endgroup
}

This section proves our key geometric theorem, which can be viewed as an anti-concentration
statement for signs of low-dimensional Gaussian projections.
Throughout, \(g\sim\mathcal{N}(0,I_d)\) denotes a standard Gaussian vector in \(\mathbb{R}^d\).
For a real number \(t\), let \(\mathrm{sgn}(t)\in\{\pm1\}\) be the sign of \(t\) (ties at \(0\) can be
broken arbitrarily since they occur with probability \(0\)).

\begin{theorem}[Geometric theorem]
\label{thm:geom}
Let \(v_1,\dots,v_n\in\mathbb{R}^d\) be unit vectors satisfying
\(\langle v_i,v_j\rangle\ge -0.9\) for all \(i\neq j\). Sample \(g\sim\mathcal{N}(0,I_d)\) and define
\(X_i=\mathrm{sgn}(\langle g,v_i\rangle)\) and \(X=\sum_{i=1}^n X_i\). Then
\[
  \E[X^2] \ \ge\ n^2\cdot 2^{-O(d)}.
\]
The hidden constant in \(O(d)\) depends only on the threshold \(-0.9\).
\end{theorem}

The statement is robust: any fixed lower bound \(\langle v_i,v_j\rangle\ge -1+\Omega(1)\) yields
\(\E[X^2]\ge n^2\cdot 2^{-O(d)}\) with constants depending only on the gap from \(-1\).
We keep \(-0.9\) for concreteness.
For completeness, we present three proofs for this theorem using different techniques. 
Proof~I first reduces the theorem to showing that an $\Theta(d)$-th term in the power series expansion of $\E[X^2]$ is large (Section~\ref{subsec:power}), which is established via two different routes of an $\eps$-net argument (Proof~I-(i), Section~\ref{subsubsec:eps}) and matrix algebra (Proof~I-(ii), Section~\ref{subsubsec:matrix}). Proof II uses the theory of Gegenbauer polynomials (Section~\ref{subsec:gegen}).

\subsection{Proof I: Power Series}
\label{subsec:power}
The first proof analyzes the power series of the $\E[X^2]$ and shows that every term is nonnegative and one term is at least $2^{-\Omega(d)}$. 
Let us begin with the standard expression for correlations of Gaussian halfspaces.

\begin{lemma}[Sheppard's formula]
\label{lem:sheppard}
Let \((Y,Z)\) be a centered bivariate normal with \(\E[Y^2]=\E[Z^2]=1\) and
\(\E[YZ]=\rho\in[-1,1]\). Then
\[
  \E\bigl[\mathrm{sgn}(Y)\,\mathrm{sgn}(Z)\bigr] = \frac{2}{\pi}\arcsin(\rho).
\]
\end{lemma}
Applying it with \(Y=\langle g,v_i\rangle\) and \(Z=\langle g,v_j\rangle\) gives
\begin{equation}
\label{eq:ex2-arcsin}
  \E[X^2] = \sum_{i,j}\E[X_iX_j] = \frac{2}{\pi}\sum_{i,j=1}^n \arcsin\bigl(\langle v_i,v_j\rangle\bigr).
\end{equation}
We will also repeatedly use the Taylor series
\begin{equation}
\label{eq:arcsin-series}
  \arcsin(t)=\sum_{k=0}^{\infty} c_k\,t^{2k+1},\qquad
  c_k \coloneqq  \frac{(2k)!}{4^k\,(k!)^2\,(2k+1)} > 0.
\end{equation}
A standard Stirling estimate gives \(c_k=\Theta(k^{-3/2})\).
In particular, there exists an absolute constant \(c>0\) such that
\begin{equation}
\label{eq:ck-lower}
  c_k \ge \frac{c}{(k+1)^{3/2}}\qquad\text{for all }k\ge 0.
\end{equation}

Let $t_{ij}\coloneqq \langle v_i,v_j\rangle\in[-0.9,1]$.
Using~\eqref{eq:arcsin-series}, we may expand termwise:
\[
\E[X^2]
= \sum_{i,j=1}^n \sum_{k=0}^\infty a_k\, t_{ij}^{2k+1}
= \sum_{k=0}^\infty a_k \underbrace{\sum_{i,j=1}^n t_{ij}^{2k+1}}_{=:S_{2k+1}}.
\]
Note that for any $k \geq 0$, 
\[
S_{2k+1} = \sum_{i,j} \langle v_i^{\otimes 2k+1}, v_j^{\otimes 2k+1} \rangle 
= 
\Big\| \sum_i v_i^{\otimes 2k+1} \Big\|_2^2 \geq 0.
\]
Since $a_k>0$ for any $k \geq 0$, hence for any fixed $k$,
\[
\E[X^2]\ \ge\ a_k\, S_{2k+1}.
\]
Therefore, Theorem~\ref{thm:geom} follows from the following lemma.

\begin{lemma}
\label{lem:power}
There exists an odd integer $p = 2k+1 = \Theta(d)$ with $S_p \geq n^2 \cdot 2^{-\Theta(d)}$. 
\end{lemma}

\subsubsection{Proof I-(i): $\eps$-Net-Based Proof}
\label{subsubsec:eps}
One way to prove Lemma~\ref{lem:power} is a simple net-based argument, showing that there will be a large set of vectors whose pairwise inner products are close to $1$, so that the sum of these inner products raised to the $p$-th power dominates the sum of the $p$-th power of the other inner products, which is at most $(-0.9)^p$. 

To formalize this idea, fix $\varepsilon\coloneqq 0.1$, and set
\[
\tau \coloneqq  1-2\varepsilon^2 = 0.98,
\qquad
M \coloneqq  (1+2/\varepsilon)^d = 21^d.
\]
A standard $\eps$-net argument guarantees the existence of an $\eps$-net $N \subseteq \mathbb{S}^{d-1}$ with $|N| \le M$ so that every $x \in \mathbb{S}^{d-1}$ has $y \in N$ with $\| x - y \|_2 \leq \eps$. By the pigeonhole principle, there exists $y \in N$ such that the set
\[
T \coloneqq  \bigl\{ i \in \{1,\dots,n\} : \| y - v_i \|_2 \leq \eps \bigr\}
\]
has size $m \coloneqq  |T| \geq n/M$. By triangle inequality, any $i, j \in T$ satisfy:
\[
\| v_i - v_j \|_2 \leq 2\eps \iff \langle v_i, v_j \rangle \geq 1 - \frac{(2\eps)^2}{2} = 1 - 2\eps^2 = \tau.
\]

Let $p$ be an odd integer (to be set shortly). Then
\[
S_p = \sum_{i,j=1}^n t_{ij}^p
\ \ge\ \sum_{i,j\in T} t_{ij}^p \ +\ \sum_{\substack{(i,j)\notin T\times T}} t_{ij}^p.
\]
For $i,j\in T$, $t_{ij}\ge\tau$, so $t_{ij}^p\ge\tau^p$, hence
\[
\sum_{i,j\in T} t_{ij}^p \ge m^2 \tau^p.
\]
For $(i,j)\notin T\times T$, we only use $t_{ij}\ge -0.9$. Since $p$ is odd,
\[
t_{ij}^p \ge (-0.9)^p = -0.9^p,
\]
so
\[
\sum_{(i,j)\notin T\times T} t_{ij}^p \ge -(n^2-m^2)\,0.9^p \ge -n^2\,0.9^p.
\]
Thus
\[
S_p \ge m^2\tau^p - n^2 0.9^p \ge \frac{n^2}{M^2}\tau^p - n^2 0.9^p
= n^2\Bigl(\frac{\tau^p}{M^2} - 0.9^p\Bigr).
\]
We now choose $p$ so that $\tau^p/M^2 \ge 2\cdot 0.9^p$, i.e.
\[
\Bigl(\frac{\tau}{0.9}\Bigr)^p \ge 2M^2 = 2\cdot 21^{2d}.
\]
Since $\tau/0.9 >1$, this is possible with $p=\Theta(d)$. Concretely, define
\[
c \coloneqq  \frac{\ln(2\cdot 21^2)}{\ln(\tau/0.9)} > 0,
\]
and set
\[
p \coloneqq  2\lceil c d\rceil + 1.
\]
Since this choice of $p$ implies $\tau^p/M^2 \ge 2\cdot 0.9^p$, the previous bound yields
\[
S_p \ge n^2(2\cdot 0.9^p - 0.9^p) = n^2\cdot 0.9^p,
\]
which proves Lemma~\ref{lem:power} and Theorem~\ref{thm:geom}.

\subsubsection{Proof I-(ii): Matrix-Algebraic Proof}
\label{subsubsec:matrix}
Another way to prove Lemma~\ref{lem:power} is based on the following linear algebraic lemma.

\begin{lemma}[A PSD sum lower bound from rank and entry lower bound]\label{lem:psd-sum}
Let $A$ be an $n\times n$ real symmetric positive semidefinite matrix with
$A_{ii}=1$ for all $i$, $\mathrm{rank}(A)\le D$, and $A_{ij}\ge -\delta$ for all
$i\ne j$, where $\delta\ge 0$.
If $\delta D\le \tfrac12$, then
\[
\sum_{i,j=1}^n A_{ij}\ \ge\ \frac{n^2}{2(D+1)}.
\]
\end{lemma}

\begin{proof}
Let $J$ be the all-ones matrix and define $B\coloneqq A+\delta J$.
Since $A\succeq 0$ and $J\succeq 0$, we have $B\succeq 0$.
Also $\mathrm{rank}(B)\le \mathrm{rank}(A)+\mathrm{rank}(J)\le D+1$.

Entrywise, $B_{ii}=1+\delta$, and for $i\ne j$,
\[
B_{ij}=A_{ij}+\delta\ \ge\ 0.
\]
Moreover, for a PSD matrix with $A_{ii}=1$ we have
$|A_{ij}|\le \sqrt{A_{ii}A_{jj}}=1$ by Cauchy--Schwarz, hence $A_{ij}\le 1$
and therefore $B_{ij}\le 1+\delta$.
Thus for every entry $b\coloneqq B_{ij}\in[0,1+\delta]$ we have $b^2\le (1+\delta)b$,
and summing gives
\[
\|B\|_F^2=\sum_{i,j} B_{ij}^2 \ \le\ (1+\delta)\sum_{i,j} B_{ij}.
\]

Let the nonzero eigenvalues of $B$ be $\lambda_1,\dots,\lambda_r$,
where $r=\mathrm{rank}(B)\le D+1$.
Then
\[
\mathrm{tr}(B)=\sum_{k=1}^r \lambda_k,\qquad \|B\|_F^2=\sum_{k=1}^r \lambda_k^2.
\]
By Cauchy--Schwarz,
\[
(\mathrm{tr}(B))^2
=\Big(\sum_{k=1}^r \lambda_k\Big)^2
\le r\sum_{k=1}^r \lambda_k^2
= r\|B\|_F^2
\le (D+1)\|B\|_F^2.
\]
Hence
\[
\|B\|_F^2\ge \frac{(\mathrm{tr}(B))^2}{D+1}.
\]
Since $\mathrm{tr}(B)=\sum_i(1+\delta)=n(1+\delta)$, combining with the
previous inequality yields
\[
(1+\delta)\sum_{i,j} B_{ij}\ \ge\ \frac{n^2(1+\delta)^2}{D+1},
\qquad\text{so}\qquad
\sum_{i,j} B_{ij}\ \ge\ \frac{n^2(1+\delta)}{D+1}.
\]
Finally, $\sum_{i,j} B_{ij}=\sum_{i,j} A_{ij}+\delta n^2$, so
\[
\sum_{i,j} A_{ij}
\ge \frac{n^2(1+\delta)}{D+1}-\delta n^2
= n^2\cdot\frac{1-\delta D}{D+1}
\ge \frac{n^2}{2(D+1)},
\]
using $\delta D\le \tfrac12$.
\end{proof}

\noindent We also use the following lemma that bounds the rank of Gram matrices of tensor powers of vectors.

\begin{lemma}[Rank bound for Hadamard powers of Gram matrices]
\label{lem:hadamard-rank}
Let $G$ be the Gram matrix $G_{ij}=\langle v_i,v_j\rangle$ for vectors
$v_i\in\mathbb{R}^d$.
For any integer $p\ge 1$, the matrix $G^{\circ p}$ with entries
$(G^{\circ p})_{ij}=\langle v_i,v_j\rangle^p$ is positive semidefinite and has
rank at most $\binom{d+p-1}{p}$.
\end{lemma}

\begin{proof}
$G^{\circ p}$ is the Gram matrix of $\{w_i\}$ defined as $w_i \coloneqq  v_i^{\otimes p}$, hence PSD.
Each $w_i$, viewed as an order-$p$ tensor, is a symmetric tensor (i.e., $w_i[a_1, \dots, a_p] = w_i[\sigma(a_1), \dots, \sigma(a_p)]$ for any $a_1, \dots, a_p \in [d]$ and any permutation $\sigma : [d] \to [d]$), 
whose dimension is $\binom{d+p-1}{p}$.
Therefore the rank of the Gram matrix is at most that dimension.
\end{proof}

With these tools, we present the proof. Choose the odd integer
\[
p\coloneqq 100d+1.
\]
Let $A$ be the matrix $A_{ij}=t_{ij}^p$.
By Lemma~\ref{lem:hadamard-rank}, $A\succeq 0$ and
\[
\mathrm{rank}(A)\ \le\ D\coloneqq \binom{d+p-1}{p}
=\binom{101d}{100d+1}
\le \binom{101d}{100d}.
\]
Also $A_{ii}=1$. For $i\ne j$, since $t_{ij}\ge -0.9$ and $p$ is odd,
\[
A_{ij}=t_{ij}^p\ \ge\ (-0.9)^p=-0.9^p.
\]
We will apply Lemma~\ref{lem:psd-sum} with $\delta\coloneqq 0.9^p$.
To do so, we need $\delta D\le \tfrac12$.

For upper bounding $D$, 
we use the inequality
\[
\binom{N}{K}\le \frac{N^N}{K^K (N-K)^{N-K}}
\qquad (0\le K\le N),
\]
which gives 
\[
\binom{101d}{100d}
\le \frac{(101d)^{101d}}{(100d)^{100d}(d)^d}
=\Big(\frac{101^{101}}{100^{100}}\Big)^d
= \bigg(101\Big(1+\frac1{100}\Big)^{100}\bigg)^d
\le (101e)^d.
\]

Since $p=100d+1$, we have $\delta=0.9^p\le 0.9^{100d}=(0.9^{100})^d$.
Also $(1-x)\le e^{-x}$ gives $0.9^{100}=(1-0.1)^{100}\le e^{-10}$, hence
\[
\delta D\ \le\ (e^{-10})^d\,(101e)^d=(101e^{-9})^d.
\]
Since $e>2$, we have $e^9>2^9=512$ and thus $101e^{-9}<101/512<1/2$.

Therefore $\delta D<\tfrac12$, and Lemma~\ref{lem:psd-sum} yields
\[
S_p = \sum_{i,j=1}^n t_{ij}^p
=\sum_{i,j=1}^n A_{ij}
\ge \frac{n^2}{2(D+1)}
\ge \frac{n^2}{4D}.
\]
Using $D\le (101e)^d<512^d=2^{9d}$ gives
\[
S_p \ \ge\ \frac{n^2}{4\cdot 2^{9d}}
=\frac{n^2}{2^{9d+2}},
\]
proving Lemma~\ref{lem:power} and Theorem~\ref{thm:geom}.

\subsection{Proof II: Using Gegenbauer polynomials}
\label{subsec:gegen}
\paragraph*{High-level strategy: expand over harmonic basis.}
Let $X = \sum_{i=1}^n \operatorname{sgn}(\langle g, v_i \rangle)$. By linearity of expectation and Sheppard's formula for the signs of Gaussian random variables, we have:
\[
    \mathbb{E}[X^2] = \sum_{i,j=1}^n \mathbb{E}[\operatorname{sgn}(\langle g, v_i \rangle)\operatorname{sgn}(\langle g, v_j \rangle)] = \frac{2}{\pi} \sum_{i,j=1}^n \arcsin(\langle v_i, v_j \rangle)
\]
We are given $\langle v_i, v_j \rangle \ge -0.9$ for all $i \neq j$. 

To lower-bound this sum, we first expand $\arcsin(t)$ over the harmonic basis, similar to the previous work~\cite{woodruff2026accelerating}.
Intuitively, harmonic basis is to spherical functions as Fourier basis is to functions on the circle $\mathbb{S}^1$; every suitably-nice function $F$ (say, square-integrable) on the sphere can be decomposed into weighted sums of (spherical) harmonic polynomials $\set{H_k^m}_{m\in [N_k]}$\footnote{The degree parameter $k$ is  called the \emph{band index}, and $N_k = \smash{{d+k-1 \choose k} - {d+k-3 \choose k-2}}$ is the dimension of the degree-$k$ spherical harmonics on $\mathbb{S}^{d-1}$.}, where the weight coefficient of $H_k^m$ can be obtained by taking inner product $\seq{F, H_k^m}$ in $L_2$, the space of square-integrable functions.

Now, our target $\arcsin(t)$ is not a function on the sphere $\mathbb{S}^{d-1}$, but a \emph{kernel}; more precisely, it can be written as a kernel function $\arcsin(\seq{u,v})$ for two points $u$ and $v$ on the sphere.
By the theory of spherical harmonics (more precisely, the fact that any function $f$ continuous on $[-1,1]$
corresponds to a \emph{zonal function} $F_v(\cdot) \coloneqq f(\seq{\cdot,v})$ that is invariant under all rotations fixing $v$ on the sphere), $f$ can be expanded over an orthogonal basis of the zonal functions, known as the \emph{Gegenbauer polynomials}.
Let \EMPH{$P_{k}^{d-1}(t)$} be the \EMPH{Gegenbauer polynomial} of degree $k$ for $\mathbb{S}^{d-1}$, normalized such that $P_{k}^{d-1}(1) =~1$.\footnote{Also known as the \emph{ultraspherical polynomial}.  Notice that in some of the literature the polynomial is defined for sphere $\mathbb{S}^d$ in dimension $d$; in this paper we consistently use \EMPH{$P_{k}(t)$} to denote the Gegenbauer polynomial (of degree $k$) for $\mathbb{S}^{d-1}$.
We refer to the notes by Jean Gallier~\cite{Gallier} and the references within for an introduction.} 
%
These polynomials can be defined using \emph{Rodrigues' formula}~\cite[Proposition 1.23]{Gallier}: let $w \coloneqq \frac{d-3}{2}$.
\[
P_{k}^{d-1}(t) = \frac{(-1)^k}{2^k} \cdot \frac{\Gamma((d-1)/2)}{\Gamma(k+(d-1)/2)} \cdot \frac{1}{(1-t^2)^{w}} \cdot \frac{\partial^k}{\partial t^k} (1-t^2)^{ k+{w} }. 
\]
Throughout the section our sphere is of dimension $d-1$ for $d \ge 3$ and thus $w \coloneqq \frac{d-3}{2} \ge 0$, and define $\EMPH{$P_k(t)$} \coloneqq P_k^{d-1}(t)$.
For constant $d$, the Gegenbauer polynomial $P_{k}$ is always a polynomial in $t$ of degree $k$.

\medskip
Our goal is to construct another continuous function $Q(t) \coloneqq \sum_{k=0}^\infty q_k P_k(t)$ that lower-bounds $\arcsin(t)$.
Intuitively, the reason why working with the Gegenbauer basis gives more information than the polynomial basis (Taylor series) is that the corresponding Gegenbauer coefficients $\set{q_k}$ are always nonnegative, and every Gegenbauer polynomial $P_k$ itself is a kernel function 
%
and satisfies the addition formula \cite[Proposition~1.18]{Gallier}.  This implies that $P_k$ as a kernel is positive semi-definite, and the sum over all pairs $(v_i, v_j)$ is non-negative: 
\[
\sum_{i,j} P_k(\seq{v_i, v_j}) 
\propto \sum_{i,j} \sum_{m=1}^{N_k} H_k^m(v_i) H_k^m(v_j) 
= \sum_{m=1}^{N_k} \Paren{ \sum_{i} H_k^m(v_i) }^2
\ge 0. 
\tag{$\star$}
\]

%

\paragraph{Reduction to a lower-bound function.} 
\noindent Suppose we can construct a continuous lower-bound function $\EMPH{$Q(t)$} \coloneqq \sum_{k=0}^\infty q_k P_k(t)$ satisfying:
\begin{enumerate}[label=(\arabic*)]
    \item $Q(t) \le \arcsin(t)$ for all $t \in [-0.9, 1]$.
    \item $q_k \ge 0$ for all $k \ge 1$.
    \item $q_0 \ge 2^{-O(d)}$.
\end{enumerate}
For any set of vectors $v_i \in \mathbb{R}^d$, let $\EMPH{$\rho_{ij}$} \coloneqq \langle v_i, v_j \rangle$.
Then, evaluating the sum yields:
\begin{align*}
    \frac{\pi}{2n^2}\mathbb{E}[X^2] &= \frac{1}{n^2}\sum_{i,j=1}^n \arcsin(\rho_{ij}) \\
    &\overset{(1)}{\ge} \frac{1}{n^2}\sum_{i,j=1}^n Q(\rho_{ij}) \\
    &= \frac{1}{n^2}\sum_{i,j=1}^n \sum_{k=1}^\infty q_k P_k(\rho_{ij}) \\
    &= \smash{q_0 + \sum_{k=1}^\infty q_k \frac{1}{n^2} \sum_{i,j} P_k(\rho_{ij})} \\
    &\overset{\mathclap{ \text{($2$),($\star$)} }}{\ge} q_0 \\
    &\overset{(3)}{\ge} 2^{-O(d)}.
\end{align*}
Our goal reduces to rigorously constructing a valid $Q(t)$, satisfying (1), (2), and (3).

\subsubsection{Construction of the lower-bound function \boldmath{$Q(t)$}}
We define $Q(t)$ by subtracting a highly-concentrated polynomial penalty from $\arcsin(t)$.
Let $\EMPH{$\Pnl(t)$} \coloneqq (t+0.9)(1-t)^{10d}$; set
\[
    \EMPH{$Q(t)$} \coloneqq \arcsin(t) - C\cdot \Pnl(t),
\]
where $C > 0$ is a constant to be chosen later. 
Both the $\arcsin$ function and the polynomial penalty are continuous on $[-1,1]$ so $Q(t)$ can be expanded over the Gegenbauer basis: $Q(t)=\sum_{k=0}^\infty q_k P_k(t)$.
%
%
Since $(t+0.9) \ge 0$ for all $t \in [-0.9, 1]$ and $C > 0$, the penalty term is strictly non-negative on the domain of our inner products, guaranteeing $Q(t) \le \arcsin(t)$ on $[-0.9, 1]$ and thus satisfying (1).

\paragraph*{Selecting \boldmath{$C$} to guarantee \boldmath{$q_k \ge 0$}.}
Now we show that $Q(t)$ satisfies (2).
We expand $Q(t)$ over the spherical harmonic basis using Gegenbauer polynomials.
The Gegenbauer coefficients of $\Pnl(t)$ can be obtained by taking ``inner-product'' with $P_k(t)$: $\displaystyle\EMPH{$\pnl_k$} = \int_{-1}^1 \Pnl(t) P_k(t) \, d\sigma(t)$, where $d\sigma(t) = c_d (1-t^2)^w dt$ is the normalized spherical measure of $\mathbb{S}^{d-1}$ on $[-1, 1]$. 
Thus the Gegenbauer coefficients of $Q(t)$ are $q_k = c_k - C \pnl_k$, where \EMPH{$c_k$} are the Gegenbauer coefficients of $\arcsin(t)$; because $\arcsin(t)$ is a PSD kernel, $c_k \ge 0$.

\medskip
We prove the following useful property on the Gegenbauer coefficients of $\Pnl(t)$ in Section~\ref{sec:g-coefficients}.
\begin{restatable}{lemma}{qcoefficients}
\label{lem:q_coefficients}
Let $d \ge 3$.
For every even $k$, $\pnl_k \le 0$.
For every odd $k$, $\pnl_k \ge 0$.
Also, for any nonnegative integer $k \le 10d+1$, the above inequalities are strict; that is, $\pnl_k$ is a constant (depending on $d$) bounded away from zero.  In particular, $|\pnl_0| \ge C_0^d$ for some $C_0 > 1$.
\end{restatable}

\medskip
\noindent Now we proceed to show that $q_k \ge 0$ for every $k\ge 1$ using case analysis.
\begin{itemize}
    \item \textit{Every $k > 10d+1$}: $\pnl_k = 0$, and thus $q_k = c_k \ge 0$. This is because $\Pnl(t)$ is a polynomial of degree $10d+1$;
    every $P_k(t)$ is a degree-$k$ polynomial, so switching from standard polynomial basis to Gegenbauer basis shows that every $\pnl_k$ at degree above $10d+1$ must vanish.
    
    \item \textit{Even $k$ between $2$ and $10d$:}  Because $\arcsin(t)$ is an odd function with purely positive Taylor coefficients, $c_k = 0$ for even $k$.
    By Lemma~\ref{lem:q_coefficients}, $\pnl_k < 0$, which implies $q_k = C |\pnl_k| \ge 0$ for any $C>0$.
    
    \item \textit{Odd $k$ between $1$ and $10d+1$:} Again by $\arcsin(t)$ being an odd function and Lemma~\ref{lem:q_coefficients}, $c_k > 0$ and $\pnl_k > 0$.  We require $q_k = c_k - C \pnl_k \ge 0$. 

\end{itemize}

We can satisfy all constraints by setting 
\[
    C \coloneqq \min_{\substack{1 \le k \le 10d+1 \\ k \text{ odd}}} \frac{c_k}{|\pnl_k|}.
\]
Because this is a minimum over a finite set of strictly positive values, $C > 0$. With this carefully tuned $C$, $q_k \ge 0$ strictly holds for all $k \ge 1$, which proves that $Q(t)$ satisfies (2).

\paragraph*{Bounding the constant \boldmath{$q_0$}.}
Lastly, let's show that $Q(t)$ satisfies (3).
Since $c_0 = 0$, By Lemma~\ref{lem:q_coefficients} we have $q_0 = - C \pnl_0 = C |\pnl_0|$. 
We lower bound $C$ and $\pnl_0$ in terms of $d$:
\begin{itemize}
    \item \textit{Bounding $|\pnl_k|$:} 
    Since $|P_k(t)| \le 1$ universally on $[-1, 1]$ by the normalization, we have 
    \[ 
    |\pnl_k| \le \max_{t \in [-1, 1]} |\Pnl(t)| \le 2(2)^{10d} = 2^{O(d)}. 
    \]
    
    \item \textit{Bounding\, $c_k$:} 
    Expanding $t^k$ in the Gegenbauer basis yields $1\cdot t^k = \beta\cdot P_k(t) + [\text{lower degree terms}]$.
    Chebyshev's extremal polynomial theorem~\cite{szego-orthogonal, huang1992, rivlin} says that any degree-$k$ polynomial bounded by $1$ on $[-1, 1]$ has a leading coefficient at most $2^{k-1}$.  If we apply Chebyshev's theorem on $P_k(t)$ which is bounded by $1$ on $[-1, 1]$ and compare the coefficients, we get $\beta \ge 2^{1-k}$.
    Now $c_k$, the $k$-th Gegenbauer coefficient of $\arcsin(t)$, is just $\beta$ times the $k$-th Taylor coefficient of $\arcsin(t)$, which is at least $\Theta(k^{-3/2})$.
    (see Eq.~\ref{eq:arcsin-series}).
    This gives $c_k \ge \Theta(2^{-k+1}\cdot k^{-3/2})$. 
    For $k \le 10d+1$, $c_k \ge 2^{-O(d)}$. 
    
    \item Therefore, $\displaystyle C = \min_{\substack{1 \le k \le 10d+1 \\ k \text{ odd}}} (c_k / |\pnl_k|) \ge 2^{-O(d)} / 2^{O(d)} = 2^{-O(d)}$.
    
    \item \textit{Bounding $|\pnl_0|$:}  By Lemma~\ref{lem:q_coefficients}, $|\pnl_0| \ge C_0^d$ for some $C_0 > 1$. 
\end{itemize}
Combining these results, $q_0 = C |\pnl_0| \ge 2^{-O(d)} \cdot C_0^d \ge 2^{-O(d)}$.  This established property (3) for~$Q(t)$.


\subsubsection{Analyzing the Gegenbauer coefficients \boldmath{$\delta_k$}: Proof of Lemma~\ref{lem:q_coefficients}}
\label{sec:g-coefficients}

Recall $\pnl(t) = (t+0.9)(1-t)^{10d} = 1.9(1-t)^{10d} - (1-t)^{10d+1}$.
Define
\[
\EMPH{$I_k(A)$} \coloneqq \int_{-1}^1 (1-t)^A P_k(t) \, d\sigma(t).
\]
We can split the integration by linearity:
\[
    \pnl_k = 1.9 I_k(10d) - I_k(10d+1).
\]
Immediately we see the ratio $\frac{I_k(10d+1)}{I_k(10d)} > 1.9$ if and only if $\pnl_k < 0$.
Our goal is to show that for every $k \le A$, (1) $(-1)^{k}I_k(A) > 0$, and (2) $\frac{|I_k(A+1)|}{|I_k(A)|} \ge 1+ \frac{A}{A+d-1}$.

\paragraph*{Exact evaluation of \boldmath{$I_k(A)$}.}
Using Rodrigues' formula for Gegenbauer polynomials, 
\[
P_k(t)=
\frac{(-1)^k}{2^k (w+1)^{\overline{k}}}
(1-t^2)^{-w} \cdot
\frac{d^k}{dt^k}\Bigl((1-t^2)^{k+w}\Bigr),
\]
where $\EMPH{$z^{\overline{k}}$} \coloneqq z\cdot\ldots\cdot (z+k-1)$ is the \emph{$k$-th rising power} of $z$, and $w \coloneqq \frac{d-3}{2}$.
The weight function for $S^{d-1}$ is proportional to $(1-t^2)^w$. 
Plug in $I_k(A)$ by multiplying $P_k(t)$ with \((1-t)^A d\sigma(t) = c_d(1-t)^{A}(1-t^2)^w dt\) where $c_{d} = \frac{\Gamma(d/2)}{\sqrt{\pi}\,\Gamma((d-1)/2)} > 0$ and $\Gamma(z)$ is the Gamma function where $\Gamma(n+1) = n!$ for every integer $n$. 
The terms $(1-t^2)^{-w}$ and $(1-t^2)^w$ cancel out:
\[
I_k(A)=
c_d\cdot \frac{(-1)^k}{2^k (w+1)^{\overline{k}}}
\int_{-1}^{1}(1-t)^A \cdot
\frac{d^k}{dt^k}\Bigl((1-t^2)^{k+w}\Bigr) dt.
\]
Integrating by parts \(k\) times gives
\[
I_k(A)=c_d\cdot 
\frac{1}{2^k (w+1)^{\overline{k}}}
\int_{-1}^{1}\frac{d^k}{dt^k}\Paren{\Big. (1-t)^A}
(1-t^2)^{k+w} dt.
\]
The boundary terms vanish because $(1-t^2)^{k+w}$ has a zero of order $k+w$ at both $t = \pm 1$. (This requires $w \ge 0$ and thus $d \ge 3$.)



\medskip
\noindent Two vital properties emerge:
\begin{enumerate}[label=\arabic*.]
    \item For $0\le k\le A$,
    since 
    \[
    \frac{d^k}{ dt^k}(1-t)^A
    =
    (-1)^k (A-k+1)^{\overline{k}} (1-t)^{A-k}
    \]
    remains positive,
    so the residual integral is also strictly positive.
    Plug in the $k$-th derivative of $(1-t)^A$ into the formula of $I_k(A)$, we have 
    \[
    I_k(A)=c_d\cdot 
    \frac{(-1)^k (A-k+1)^{\overline{k}}}{2^k (w+1)^{\overline{k}}}
    \int_{-1}^{1} (1-t)^{A+w}(1+t)^{k+w} dt.
    \]
    The sign of $I_k(A)$ is $(-1)^k$. 
    Therefore, $I_k(A) > 0$ for even $k$, and $I_k(A) < 0$ for odd $k$ for every $k\le A$.
    If $k>A$, \(\frac{d^k}{dt^k}(1-t)^A=0\), hence \(I_k(A)=0\).
    
    \item 

    The ratio of $I_k(A+1)$ and $I_k(A)$ has a nice closed form:  Again assume $0\le k\le A$.
    Observe that the integral in $I_k(A)$ can be written as $2^{A+k+2w}$ times a Beta function
    \(
    B(A+w+1, k+w+1) 
    \)
    where $B(x,y) = \frac{\Gamma(x)\Gamma(y)}{\Gamma(x+y)} = \int_0^1 u^{x-1}(1-u)^{y-1} \,du$, by setting $u=\frac{1+t}{2}$.
    The Beta function satisfies a nice ratio identity $B(x+1,y)/B(x,y) = x/(x+y)$:
    \[
    \frac{B(A+w+2, k+w+1)}{B(A+w+1, k+w+1)}
    =
    \frac{A+w+1}{A+k+2w+2}.
    \]

    Therefor we obtain, for \(0\le k\le A\),
    \begin{align}
    \label{eq:Ik-ratio}
        \frac{|I_k(A+1)|}{|I_k(A)|} = 2 \left( \frac{A+1}{A+1-k} \right) \left( \frac{A+w+1}{A+k+2w+2} \right).
    \end{align}
    For $k \ge 0$, the denominator function $f(k) = (A+1-k)(A+k+d-1)$ 
    is strictly decreasing. 
    Therefore, the ratio is strictly increasing with $k$, achieving its global minimum on $[0, A]$ at $k=0$.
    This implies
    \[
        \frac{|I_k(A+1)|}{|I_k(A)|} \ge 2 \left( \frac{A+w+1}{A+2w+2} \right)
        \ge 1+\frac{A}{A+2w+2}.
    \]
\end{enumerate}

\paragraph*{Sign Alternation of \boldmath{$\pnl_k$}.}
We substitute $A = 10d$ into the ratio:
\[
    \min_{0 \le k \le A} \frac{|I_k(10d+1)|}{|I_k(10d)|} \ge 1+ \frac{10d}{11d-1} = \frac{21d - 1}{11d - 1}.
\]
For any dimension $d \ge 2$, this ratio evaluates to at least $\frac{41}{21} \approx 1.952 > 1.9$.

\medskip
We can now cleanly determine the signs of $\pnl_k = 1.9 I_k(10d) - I_k(10d+1)$ for all $k \le 10d+1$:
\begin{itemize}
    \item \textit{For even $k$:} Both $I_k(10d)$ and $I_k(10d+1)$ are strictly positive. Since the ratio of their absolute values $\ge 1.952$, we have $I_k(10d+1) \ge 1.952 I_k(10d) > 1.9 I_k(10d)$, forcing \textbf{$\pnl_k < 0$}.
    \item \textit{For odd $k$:} Both $I_k(10d)$ and $I_k(10d+1)$ are strictly negative. Because the ratio of their absolute values $\ge 1.952$, the magnitude of the positive term $-I_k(10d+1)$ thoroughly dominates, forcing \textbf{$\pnl_k > 0$}.
    \item \textit{When $k = 10d+1$:} $I_k(10d) = 0$ and thus the parity of $\pnl_k$ is solely determined by $I_k(10d+1)$. 
    \textit{(Note: For $k > 10d+1$, the maximum degree of $\Pnl(t)$ is exceeded, dictating $\pnl_k = 0$ exactly).}
\end{itemize}

\paragraph*{Lower bounds for \boldmath{$\pnl_0$}.}
From the lower bound for the ratio $\frac{|I_k(10d+1)|}{|I_k(10d)|}$ and the exact formula of $|I_k(10d)|$ using Beta function, by substituting $A=10d$ we get 
\begin{align*}
    |\pnl_0| = |1.9 I_0(10d) - I_0(10d)|
    &\ge \Paren{\frac{21d - 1}{11d - 1} - \frac{19}{10}} \cdot I_k(10d) \\
    &= \Paren{\frac{21d - 1}{11d - 1} - \frac{19}{10}} \cdot 2^{10d+w} \cdot B(10d+w+1, w+1) \\
    &\ge \frac{d+9}{10(11d - 1)} \Paren{\frac{3}{2}}^{9d+1} \\
    &\ge \frac{1}{10} \Paren{\frac{3}{2}}^{9d+1} 
\end{align*}
which shows $|\pnl_0| \ge C_0^d$ for some $C_0 > 1$.

\medskip
This concludes the proof of Lemma~\ref{lem:q_coefficients}.

\subsection{Upper Bound on the Variance}
The $2^{-O(d)} n^2$ lower bound on $\E[X]$ is asymptotically tight, as shown in the following proposition. 

\begin{proposition}
There exist infinitely many integers $d$ for which the following holds.
Let
\[
n \coloneqq \Floor{e^{d/100}}.
\]
Then there exist unit vectors $v_1,\dots,v_n \in \mathbb R^d$ such that $\langle v_i, v_j\rangle \ge -0.9$ for all $i\neq j$, so that if $g\sim \mathcal N(0,I_d)$ and $
X_i \coloneqq \operatorname{sgn}(\langle g, v_i\rangle)\in\{\pm 1\}$ and 
$X \coloneqq \sum_{i=1}^n X_i$, 
then $\mathbb E_g[X^2] \le 2n$.
In particular,
\[
\frac{\mathbb E_g[X^2]}{n^2} \le \frac{2}{n} \le 2e^{-d/100} = e^{-\Omega(d)}.
\]
\end{proposition}

\begin{proof}
We use the probabilistic method. Sample $v_1,\dots,v_n$ i.i.d.\ uniformly at random from the unit sphere $\mathbb{S}^{d-1}$.

\medskip
\noindent\textbf{1. Ensuring $\langle v _i, v_j \rangle \geq -0.9$.}
Fix $i\neq j$. By symmetry,
\[
\mathbb P(\langle v_i,v_j\rangle < -0.9)=\mathbb P(\langle v_i,v_j\rangle > 0.9).
\]
We use the following standard claim that shows a tail bound on the inner product of two random $d$-dimensional unit vectors. For completeness, we provide a proof at the of this subsection.

\begin{claim}
For independent uniform $U,V\in \mathbb{S}^{d-1}$ and any $a\in(0,1)$,
\begin{equation}\label{eq:cap}
\mathbb P(\langle U,V\rangle \ge a) \le (1-a^2)^{\frac{d-1}{2}}.
\end{equation}
\label{claim:unit-vector}
\end{claim}

Using the above claim with $a=0.9$ we get
\[
\mathbb P(\langle v_i,v_j\rangle < -0.9) \le 0.19^{\frac{d-1}{2}}.
\]    
Let $A$ be the event that $\langle v_i,v_j\rangle \ge -0.9$ for all $i\neq j$.
By the union bound over $\binom{n}{2}$ pairs,
\[
\mathbb P(A^c)
\le \binom{n}{2}\, 0.19^{\frac{d-1}{2}}
\le \frac{n^2}{2}\, 0.19^{\frac{d-1}{2}}.
\]
Since $n^2\le e^{d/50}$ and $0.19^{(d-1)/2}=\exp\!\big(-\tfrac{d-1}{2}\ln(1/0.19)\big)$,
there exists an absolute constant $c>0$ such that
\[
\mathbb P(A^c)\le e^{-cd}
\]
for all sufficiently large $d$. In particular, for all sufficiently large $d$,
\begin{equation}\label{eq:PA}
\mathbb P(A)\ge 0.99.
\end{equation}

\medskip
\noindent\textbf{2. Upper bounding $\E[X^2]$.}
For a given realization of $(v_1,\dots,v_n)$ define the nonnegative quantity
\[
F(v_1,\dots,v_n) \coloneqq \mathbb E_g[X^2].
\]
We compute its expectation over the random choice of $v_1,\dots,v_n$:
\[
\mathbb E_v[F] = \mathbb E_{v,g}[X^2].
\]
Expand
\[
X^2=\sum_{i=1}^n X_i^2+\sum_{i\neq j} X_iX_j.
\]
Since $X_i\in\{\pm 1\}$, we have $X_i^2=1$ always, hence
\[
\mathbb E_{v,g}\Big[\sum_{i=1}^n X_i^2\Big]=n.
\]
We claim that for each $i\neq j$,
\[
\mathbb E_{v,g}[X_iX_j]=0.
\]
Indeed, condition on $g$ and on all vectors $\{v_k\}_{k\neq i}$; under this conditioning, $v_i$ remains uniform on the sphere and is invariant under $v_i\mapsto -v_i$. But replacing $v_i$ by $-v_i$ flips the sign of $X_i=\operatorname{sgn}(\langle g,v_i\rangle)$ while leaving $X_j$ unchanged, hence the product $X_iX_j$ changes sign. Therefore its conditional expectation is $0$, and taking expectations yields $\mathbb E_{v,g}[X_iX_j]=0$.

Consequently,
\[
\mathbb E_v[F]=\mathbb E_{v,g}[X^2]=n.
\]
By Markov's inequality applied to the nonnegative random variable $F$,
\[
\mathbb P_v(F\ge 2n)\le \frac{\mathbb E_v[F]}{2n}=\frac{1}{2},
\]
so if we define $B\coloneqq\{F\le 2n\}$ then
\begin{equation}\label{eq:PB}
\mathbb P(B)\ge \frac{1}{2}.
\end{equation}

\medskip
\noindent\textbf{Step 3: positive probability intersection.}
From \eqref{eq:PA} and \eqref{eq:PB},
\[
\mathbb P(A\cap B)\ge \mathbb P(A)+\mathbb P(B)-1 \ge 0.99+\frac12-1>0.
\]
Hence there exists a deterministic choice of unit vectors $(v_1,\dots,v_n)$ such that both $A$ and $B$ hold, i.e.,
\[
\langle v_i,v_j\rangle \ge -0.9 \quad \forall i\neq j,
\qquad\text{and}\qquad
\mathbb E_g[X^2]\le 2n.
\]
Finally, for this choice,
\[
\frac{\mathbb E_g[X^2]}{n^2}\le \frac{2n}{n^2} = \frac{2}{n} \le 2e^{-d/100} = e^{-\Omega(d)}.
\]
\end{proof}

Now we prove Claim~\ref{claim:unit-vector} on the tail bound on the inner product between two independent $d$-dimensional unit vectors. 
\begin{proof}
[of Claim~\ref{claim:unit-vector}]
Fix $U$ and, by rotational invariance, assume $U=e_1$. Let $V$ be uniform on $\mathbb{S}^{d-1}$. A standard representation is
\[
V=\frac{G}{\|G\|},\qquad\text{where } G\sim \mathcal N(0,I_d).
\]
Write $G=(Z,W)$ where $Z\coloneqq G_1\sim \mathcal N(0,1)$ and $W\coloneqq (G_2,\dots,G_d)$, and let $S\coloneqq \|W\|^2=\sum_{k=2}^d G_k^2\sim\chi^2_{d-1}$, independent of $Z$.
Then
\[
\langle U,V\rangle = V_1 = \frac{Z}{\sqrt{Z^2+S}}.
\]
If $V_1\ge a$ (with $a\in(0,1)$), then $Z\ge 0$ and squaring gives
\[
\frac{Z^2}{Z^2+S}\ge a^2
\quad\Longrightarrow\quad
(1-a^2)Z^2 \ge a^2 S
\quad\Longrightarrow\quad
Z^2 \ge \lambda S,
\qquad \lambda\coloneqq \frac{a^2}{1-a^2}.
\]
Conditioning on $S$, we have
\[
\mathbb P(V_1\ge a\mid S)\le \mathbb P(Z^2\ge \lambda S\mid S)
= \mathbb P(|Z|\ge \sqrt{\lambda S}\mid S).
\]
Using the standard Gaussian tail bound $\mathbb P(|Z|\ge t)\le e^{-t^2/2}$ for $t\ge 0$,
\[
\mathbb P(V_1\ge a\mid S)\le \exp\!\left(-\frac{\lambda S}{2}\right).
\]
Taking expectation over $S$ yields
\[
\mathbb P(V_1\ge a) \le \mathbb E\left[\exp\!\left(-\frac{\lambda S}{2}\right)\right].
\]
Now $S=\sum_{k=2}^d G_k^2$ is a sum of $(d-1)$ independent $\chi^2_1$ variables; for $Y\sim\chi^2_1$ one has
$\mathbb E[e^{-tY}]=(1+2t)^{-1/2}$ for $t\ge 0$ (a one-dimensional Gaussian integral),
so
\[
\mathbb E\left[e^{-\lambda S/2}\right]
=\left(\mathbb E\left[e^{-(\lambda/2)Y}\right]\right)^{d-1}
=(1+\lambda)^{-(d-1)/2}.
\]
Since $1+\lambda = 1+\frac{a^2}{1-a^2}=\frac{1}{1-a^2}$, we conclude
\[
\mathbb P(V_1\ge a) \le (1-a^2)^{\frac{d-1}{2}}.
\]
This is exactly \eqref{eq:cap}, completing the proof.
\end{proof}

%% file: algorithms.tex

\label{sec:alg}

In this section we show how Theorem~\ref{thm:geom} yields an improved approximation for Max-Cut. The algorithm and its analysis are almost identical to the local improvement-based framework from \cite{FKL02,HK23} -- the key difference is that we rely on our geometric anti-concentration lemma to show the additional improvement on top of the $\alpha_{\mathrm{GW}}$-factor approximation guarantee for Max-Cut. We describe and analyze the algorithm in detail in this section.

\subsection{The SDP Relaxation and Rounding Scheme}

As discussed earlier, our first step would be to solve the basic SDP relaxation strengthened with the $\ell^2_2$-triangle inequalities:

\begin{figure}[ht!]
\begin{align}				\label{eq:sdp-basic}	
	\mathrm{SDP} := 
	& \max\ \frac12\sum_{(i,j)\in E} w_{ij}\bigl(1 - \langle v_i,v_j\rangle\bigr) & \\
	& \qquad \text{s.t. } \|v_i\|^2_2=1 & \forall i\in V,  \\
	& \|a_iv_i - a_jv_j \|^2_2 + \|a_iv_i - a_kv_k \|^2_2 \geq \|a_jv_j - a_kv_k\|^2_2  & \forall a_i,a_j,a_k \in \{\pm 1\}, \nonumber \\
	& & \forall i,j,k \in V. \label{eqn:triangle}			
\end{align}
\end{figure}

Clearly, the SDP relaxation is feasible, since any integral assignment $\{v_i \in \{\pm 1\}\}_{i \in V}$ satisfies all the constraints. Hereafter, we will work with a feasible SDP solution $\{v_i\}_{i \in V}$ where the vectors lie in $\mathbb{R}^d$ (or equivalently, the corresponding Gram matrix has rank at most $d$), and use $\mathrm{SDP}$ to denote the value of the corresponding solution.

\paragraph{Rounding the SDP solution.} Next, we use the rounding scheme from \cite{FKL02,HK23} to round the above SDP solution into an integral solution. The rounding algorithm, Algorithm \ref{alg:local}, has two parts. Firstly, it samples a Gaussian vector \(g\sim\mathcal{N}(0,I_d)\) and projects it along the vectors $v_1,\ldots,v_n$ and outputs the corresponding signs $x_i = \mathrm{sgn}(\langle g, v_i \rangle)$ as a preliminary set of labels. This is then followed by a greedy post-processing step where some of labels of the vertices are flipped to further improve the quality of the solution. For completeness, we describe the rounding scheme in full detail below: 

\begin{algorithm}[ht!]
\caption{GW rounding with local improvement}
\label{alg:local}
\begin{enumerate}
  \item Input: a \(d\)-dimensional SDP solution \(\{v_i\}_{i\in V}\subseteq\mathbb{R}^d\).
  \item Set \(\varepsilon := 2^{-C d}\) for a sufficiently large absolute constant \(C>0\).
  \item Sample \(g\sim\mathcal{N}(0,I_d)\) and set \(x_i \leftarrow \mathrm{sgn}(\langle g,v_i\rangle)\) for all \(i\in V\).
  \item Candidate set: \(S \leftarrow \{ i\in V : |\langle g,v_i\rangle| < \varepsilon\}\).
  \item For each \(i\in S\):
   \begin{enumerate}
    \item Define \(W_i := \sum_{j\in N(i)} w_{ij}\).
    \item Define \(B_i := \{ j\in N(i) : x_ix_j = 1\ \text{and}\ |\langle g,v_j\rangle|\ge \varepsilon\}\).
    \item If \(\sum_{j\in B_i} w_{ij} > W_i/2\), then flip \(x_i\leftarrow -x_i\).
  \end{enumerate}
  \item Output \(x\).
\end{enumerate}
\end{algorithm}

Intuitively, one can think of $S$ as the subset of vertices that are close to being undecided with respect to the fixed hyperplane, and in particular, for any vertex $i \in V$, the joint distribution over the labels of the neighborhood of $i$ remains almost unchanged even after conditioning on the label of $i$; this observation is in crucial in arguing that for the vertices in $S$, the expected gain in the value of the solution after post-processing is at least $e^{-O(d)}$ (Lemma \ref{lem:local-gain}). 

Now, we proceed to analyze the performance of the rounding algorithm. We begin with the following standard fact.

\begin{claim}	\label{cl:edge-wise}
	For any edge $(i,j) \in E$, we have
	\[
		\Pr_{g \sim \mathcal{N}(0,I_d)}\left[x_i \neq x_j\right] \geq \alpha_{GW}\cdot \frac{\|v_i - v_j\|^2_2}{4}.
	\]
    Hence, the expected number edges cut by the labeling $(x_i)^n_{i = 1}$ from Line 3 is at least $\alpha_{\mathrm{GW}}\cdot \mathrm{SDP}$.
\end{claim}

Next, we can easily verify the following:

\begin{claim}
	For any $i \in V$, the \emph{local gain} at flipping $x_i \gets -x_i$ is 
	\begin{equation}
		\label{eq:gain-def}
		\Delta_i := \Bigl( \sum_{j\in B_i} w_{ij} - \sum_{j\in N(i)\setminus B_i} w_{ij} \Bigr)_+
		= \bigl(2\sum_{j\in B_i} w_{ij} - W_i\bigr)_+.
	\end{equation}
	Hence, the total increase in the objective due to the post-processing step is at least
	\(\Delta:=\sum_{i\in S}\Delta_i\).
\end{claim}

In the following section, we will bound the expected total gain $\Delta_i$ at a given vertex $i$.

\subsection{Expected local gain}

Next, we give the following lemma (which is an analogue of Lemma 8 from \cite{HK23}) which lower bounds the expected gain $\Delta_i$ during the post-processing phase. 

\begin{lemma}[Expected local gain]
\label{lem:local-gain}
Suppose $\rho_{ij} := \langle v_i,v_j\rangle\in[\rho^*\pm 0.01]$ for every edge $(i,j) \in E$.
Then, for every vertex \(i\in V\),
\[
  \E[\Delta_i\mid i\in S] \ \ge\ \ W_i\cdot 2^{-O(d)}.
\]
\end{lemma}

To prove the above lemma, we will need the following generalization of our geometric anti-concentration theorem (Theorem \ref{thm:geom}) to the weighted setting:

\begin{corollary}[Geometric theorem - Weighted Version]
	\label{cor:weighted}
	Let \(v_1,\dots,v_n\in\mathbb{R}^d\) be unit vectors satisfying
	\(\langle v_i,v_j\rangle\ge -0.9\) for all \(i\neq j\). Let $w_1,\ldots,w_n \in \mathbb{R}_{\geq 0}$ be a collection of non-negative weights. 	
	Sample \(g\sim\mathcal{N}(0,I_d)\) and define
	\(X_i :=\mathrm{sgn}(\langle g,v_i\rangle)\) and \(X :=\sum_{i=1}^n w_i X_i\). Then,
	\[
	\E[X^2] \ \ge\ W^2\cdot 2^{-O(d)},
	\]
	where $W := \sum^n_{i = 1} w_i$ denotes the sum of the weights. Again, the hidden constant in \(O(d)\) depends only on the threshold \(-0.9\).
\end{corollary}
The weighted version stated above can be derived by a straightforward reduction to the unweighted case, and we give a proof of the corollary in Section \ref{sec:weighted}. Equipped with the above, we can now proceed to prove Lemma \ref{lem:local-gain}. 

\begin{proof}[Proof of Lemma \ref{lem:local-gain}]
We will proceed by fixing a vertex $i \in V$ and condition on \(i\in S\), i.e., \(|\langle g,v_i\rangle|<\varepsilon\).
By rotational invariance, we may assume \(v_i=e_1\) and write \(g=(g_1,g')\) with
\(g_1\sim\mathcal{N}(0,1)\) and \(g'\sim\mathcal{N}(0,I_{d-1})\) independent. For every $j \in N(i)$, let \(\rho_j:=\rho_{ij}=\langle v_i,v_j\rangle\in[\rho^*\pm 0.01]\). Then we can write
\[
 v_j = \rho_j e_1 + \sqrt{1-\rho_j^2}\,u_j,
\]
where $u_j \in \mathbb{R}^d$ is the normalized component of $v_j$ that is orthogonal to $e_1$. Observe that \(x_i=\mathrm{sgn}(g_1)\), and for each neighbor \(j\) let us define the Gaussian variable
\begin{equation}
\label{eq:hj-def}
  h_j := \mathrm{sgn}(g_1)\,\langle g',u_j\rangle.
\end{equation}
Then conditioned on \(g_1\), the vector \((h_j)_{j\in N(i)}\) is a Gaussian vector random variable in $\mathbb{R}^d$ with mean $0$ and covariance matrix
\(\Sigma\) given by
\begin{equation}
\label{eq:sigma-def}
  \Sigma_{jk} := \E[h_j h_k] = \langle u_j,u_k\rangle.
\end{equation}

\paragraph{Lower bounding correlations.}
We claim that for all distinct neighbors \(j,k\in N(i)\),
\begin{equation}
\label{eq:sigma-lb}
  \Sigma_{jk} \ge -0.2.
\end{equation}
The above can be argued as follows: note that the vectors $v_i,v_j,v_k$ satisfy \eqref{eqn:triangle} and hence,
\[
\|v_i + v_j\|^2_2 + \|v_i + v_k\|^2_2 \geq \|v_j - v_k\|^2_2,
\]
which on rearranging yields $\rho_{jk} \geq -\rho_{ij} - \rho_{ik} - 1$. On the other hand, by expanding and substituting the decomposition above yields
\[
  \rho_{jk} = \langle v_j,v_k \rangle =  \rho_j\rho_k + \sqrt{1-\rho_j^2}\,\sqrt{1-\rho_k^2}\,\langle u_j,u_k\rangle,
\]
and hence
\[
\Sigma_{jk} = \langle u_j,u_k \rangle 
= \frac{\rho_{jk} - \rho_j\rho_k}{\sqrt{1 - \rho^2_j}\sqrt{1 - \rho^2_k}}
\geq \frac{-\rho_j - \rho_k - 1 - \rho_j\rho_k}{\sqrt{1 - \rho^2_j}\sqrt{1 - \rho^2_k}},
\]
which along with the fact that $\rho_j,\rho_k \in [\rho^* \pm 0.01]$ and $\rho^* \approx -0.689$ implies that $\Sigma_{jk}  \geq - 0.2$, proving \eqref{eq:sigma-lb}.

\paragraph{A surrogate random variable.}
Now define
\begin{equation}
\label{eq:Z-def}
  Z := \sum_{j\in N(i)} w_{ij}\,\mathbf{1}[h_j \geq 3\varepsilon].
\end{equation}
If \(|g_1|<\varepsilon\) and \(h_j \geq 3\varepsilon\), then
\(\mathrm{sgn}(g_1)\langle g,v_j\rangle = \rho_j|g_1|+\sqrt{1-\rho_j^2}\,h_j \geq \varepsilon\)
(for a suitable choice of the absolute constant \(C\) in \(\varepsilon=2^{-Cd}\)).
In particular, this implies \(j\in B_i\), so
\begin{equation}
\label{eq:Z-le-B}
  Z \le \sum_{j\in B_i} w_{ij}.
\end{equation}
Combining \eqref{eq:gain-def} and \eqref{eq:Z-le-B}, we observe that
\begin{equation}
\label{eq:delta-vs-Z}
  \Delta_i \ge (2Z-W_i)_+ = 2\,(Z-W_i/2)_+,
\end{equation}
and hence it will suffice to focus on lower bounding the random variable $(Z - W_i/2)_+$.

\paragraph{Second moment computation.}
Let \(p_0:=\Pr[\mathcal{N}(0,1) \geq 3\varepsilon]=\tfrac12-\Theta(\varepsilon)\).
A standard Gaussian anti-concentration bound implies that for any correlation \(\rho\in[-1,1]\),
\begin{equation}
\label{eq:bivar-shift}
  \Pr[A \geq 3\varepsilon, B \geq 3\varepsilon]
  = \Pr[A \geq 0, B \geq 0] \pm O(\varepsilon)
  = \frac14 + \frac{1}{2\pi}\arcsin(\rho) \pm O(\varepsilon),
\end{equation}
where \((A,B)\) is a centered bivariate normal with unit variances and correlation \(\rho\), and we used
Lemma~\ref{lem:sheppard} in the second equality.
Using \eqref{eq:bivar-shift} and \eqref{eq:sigma-def}, a direct expansion gives
\begin{equation}
\label{eq:varZ}
  \E\bigl[(Z-W_i/2)^2\mid i\in S\bigr]
  \ge \frac{1}{2\pi}\sum_{j,k\in N(i)} w_{ij}w_{ik}\arcsin(\Sigma_{jk}) - O(\varepsilon W_i^2).
\end{equation}

\paragraph{Invoking the geometric lemma.}
The matrix \(\Sigma\) is a Gram matrix of vectors \(u_j\in\mathbb{R}^{d-1}\) with diagonal ones and
off-diagonal entries at least \(-0.2\) by \eqref{eq:sigma-lb}.
Applying Corollary~\ref{cor:weighted} (with a slightly weaker threshold than \(-0.9\), which only
improves the constants) yields
\begin{equation}
\label{eq:arcsin-lb}
  \sum_{j,k\in N(i)} w_{ij}w_{ik}\arcsin(\Sigma_{jk}) \ \ge\ \ W_i^2\cdot 2^{-O(d)}.
\end{equation}
Choosing the constant \(C\) in \(\varepsilon=2^{-Cd}\) large enough so that
\(\varepsilon\le 2^{-\Omega(d)}\) is dominated by the right-hand side of \eqref{eq:arcsin-lb}, we
combine \eqref{eq:varZ} and \eqref{eq:arcsin-lb} to obtain
\begin{equation}
\label{eq:varZ-lb-final}
  \E\bigl[(Z-W_i/2)^2\mid i\in S\bigr] \ \ge\ \ W_i^2\cdot 2^{-O(d)}.
\end{equation}

\paragraph{From variance to expected gain.}

Now we can finally convert the variance bound to a lower bound on the expected gain. Let us define random variables \(\widetilde{Z}:=Z-W_i/2\),  \(\widetilde{Z}_+:=\max(\widetilde{Z},0)\), and $\widetilde{Z}_- := \max(-\widetilde{Z},0)$. Note that by definition \(0\le Z\le W_i\), and hence it follows that $|\widetilde{Z}| \leq W_i/2$. Furthermore, we can also bound:
\begin{align*}
\E[\widetilde{Z}_+] - \E[\widetilde{Z}_-]
= \E[\widetilde{Z}] 
&= \E\left[\sum_{j\in N(i)} w_{ij}\,\mathbf{1}[h_j \geq 3\varepsilon]\right] - \frac{W_i}{2}	\\
&= \sum_{j\in N(i)} w_{ij}\E\left[\,\mathbf{1}[h_j \geq 3\varepsilon]\right] - \frac{W_i}{2}	\\
&= W_i\left(\frac12 - \Theta(\epsilon)\right) - \frac{W_i}{2}	\\
&= -\Theta(\epsilon)W_i,	
\end{align*}
and hence $\E[\widetilde{Z}_-] \leq \E[\widetilde{Z}_+] + O(\epsilon)W_i$. Combining the above bounds, we can now observe:
\[
\E[\widetilde{Z}^2] 
\leq (W_i/2) \E[\widetilde{Z}_+ + \widetilde{Z}_-] 
\leq W_i\widetilde{Z}_+ + O(\epsilon W^2_i) 
\]
Further, using \eqref{eq:varZ-lb-final} and by setting the parameter $C$ large enough in $\epsilon = e^{-Cd}$, we get that
\(\E[\widetilde{Z}_+\mid i\in S] \ge W_i\cdot 2^{-O(d)}\). Finally, \eqref{eq:delta-vs-Z} shows that:
\[
\E[\Delta_i\mid i\in S]=2\E[\widetilde{Z}_+\mid i\in S] \ge W_i\cdot 2^{-O(d)},
\]
which completes the proof.
\end{proof}

\subsection{From local gain to the main theorem}
\label{subsec:mainproof}
We now complete the proof of Theorem~\ref{thm:main}.

\begin{proof}[Proof of Theorem~\ref{thm:main}]
Let \(W:=\sum_{(i,j)\in E} w_{ij}\) be the total edge weight. Let \(\mathrm{SDP}\) denote the value of the
input solution \eqref{eq:sdp-basic}.

\paragraph{Case 1: constant fraction of edges are far from worst correlation.}
Let \(\eta=0.01\). Since the function
\(F(\rho)=(1-\tfrac{2}{\pi}\arcsin(\rho))/(1-\rho)\) has a unique minimizer \(\rho^*\) and is smooth,
there exists a constant \(\gamma>0\) such that
\(F(\rho)\ge \alpha_{\mathrm{GW}}+\gamma\) for all \(|\rho-\rho^*|\ge \eta\).
If the total SDP contribution of edges with \(|\rho_{ij}-\rho^*|\ge \eta\) is at least
\(\delta\cdot \mathrm{SDP}\) for $\delta = e^{-C'd}$ (for $C' > 0$ large enough), then by Claim \ref{cl:edge-wise} and the definition of
\(F\), the plain GW rounding already yields expected value at least
$(\alpha_{\mathrm{GW}}+\gamma\delta)\cdot\mathrm{SDP} = (\alpha_{\mathrm{GW}}+ e^{-O(d)})\cdot\mathrm{SDP}$.

\paragraph{Case 2: the edgewise GW analysis is tight.}
Otherwise, after discarding at most $\delta$-fraction of the SDP value, we may assume that every edge
satisfies \(\rho_{ij}\in[\rho^*\pm \eta]\), which is the hypothesis of Lemma~\ref{lem:local-gain}. 

Now note that $\mathrm{SDP}$ is a constant fraction of \(W\), because
\(\mathrm{SDP}=\tfrac12\sum_{(i,j)}w_{ij}(1 - \rho_{ij})\ge \tfrac12(1 -\rho^*-\eta)W\).
On the other hand, by Lemma~\ref{lem:local-gain} and the fact that
\(\Pr[i\in S]=\Pr[|g_1|<\varepsilon]=\Theta(\varepsilon)\), we have
\begin{align*}
  \E[\Delta]=\sum_{i\in V} \E[\mathbf{1}[i\in S]\,\Delta_i]
  = \Theta(\varepsilon)\sum_{i\in V} \E[\Delta_i\mid i\in S]
  &\ge \Theta(\varepsilon)\sum_{i\in V} W_i\cdot 2^{-O(d)} - O(\delta)\cdot W		\\
  &\geq W\cdot 2^{-O(d)},
\end{align*}
where we used \(\sum_i W_i = 2W\) and the observation that \(\varepsilon=2^{-Cd}\), the factor \(\Theta(\varepsilon)\cdot 2^{-O(d)}\) is still
\(2^{-O(d)}\), as claimed. Note that in the above computation, we remove a $O(\delta)W$-weight of edges from the cumulative local gain to account for the edges for which the correlation $\rho \notin [\rho^* \pm \eta]$. 

Finally, observe that (a) the labeling in Line $3$ of Algorithm \ref{alg:local} cuts $\alpha_{\mathrm{GW}} \cdot\mathrm{SDP}$ weight of edges in expectation (Claim \ref{cl:edge-wise}), and (b) the expected additive improvement in the number of cut edges due to Lines $4$ and $5$ is $\E[\Delta] \geq W \cdot 2^{-O(d)}$. Combining these observations we obtain the bound
\[
\E\Big[\text{weight of edges cut}\Big]\ge (\alpha_{\mathrm{GW}}+2^{-O(d)})\cdot\mathrm{SDP},
\]
which concludes the proof of Theorem \ref{thm:main}.
\end{proof}

\subsection{Proof of Corollary \ref{cor:weighted}}			\label{sec:weighted}

Recall that in this corollary, we would like to lower bound $\E[X^2]$, where $X = \sum^n_{i = 1}w_ig_i$. To that end, let $C$ denote the constant in the exponent from Theorem \ref{thm:geom}. Further, recall that we denote the sum of the weights by $W = \sum^n_{i = 1}w_i$. For the rest of the proof, let us define a parameter $\eps := (2^{-Cd}\cdot W^2)/4$. Now, given weights $w_1,\ldots,w_n$, we can choose rational weights $w'_1,\ldots,w'_n \in \mathbb{Q}$ such that\footnote{This can be done since $\mathbb{Q}$ is dense in $\mathbb{R}$ and the maps $(x_1,\ldots,x_n) \mapsto (\sum^n_{i = 1} x_i)^2$ and $(x_1,\ldots,x_n) \mapsto \E[(\sum^n_{i = 1}x_ig_i)^2]$ are continuous.}:
\begin{equation}			\label{eqn:e1}
	\left|\E\left[\left(\sum^n_{i = 1}w_ig_i\right)^2\right] - \E\left[\left(\sum^n_{i = 1}w'_ig_i\right)^2\right]\right|\leq \eps,
\end{equation}
and
\begin{equation}			\label{eqn:e2}
	\left|\left(\sum^n_{i = 1}w_i\right)^2 - \left(\sum^n_{i = 1}w'_i\right)^2\right| \leq \eps.
\end{equation}
Next, since the new weights $w'_1,\ldots,w'_n$ are rational, we can choose a positive integer $N \in \mathbb{N}$ such that $n_i = w'_i \cdot N$ is an integer for every $i \in [n]$. 

Now, consider the following collection of jointly distributed Gaussian random variables: for every $i \in [n]$, we introduce $n_i$-copies of the Gaussian random variable $g_i$, and call them $g_{i,1},\ldots,g_{i,n_i}$. Note that the resulting collection of Gaussian random variables $G' := (g_{i,j} : i \in [n], j \in [n_i])$ satisfy the conditions required to apply Theorem \ref{thm:geom}:
\begin{itemize}
	\item The rank of the covariance matrix of $G'$ is the same as that of the covariance matrix of $(g_i)_{i \in [n]}$, and hence by the setting of the corollary, the rank is at most $d$.
	\item Any pair of Gaussians $g,g' \in G'$ also satisfy $\E[gg'] \geq -0.9$.
\end{itemize}

Equipped with the above setup, we can now proceed as follows. Firstly, using \eqref{eqn:e1} and using the definition $\{n_i\}_i$, we can lower bound:
\begin{equation}		\label{eqn:e3}
	\E\left[\left(\sum^n_{i = 1}w_ig_i\right)^2\right]
	\geq \E\left[\left(\sum^n_{i = 1}w'_ig_i\right)^2\right] - \eps 
	\geq \frac{1}{N^2}\E\left[\left(\sum^n_{i = 1}n_ig_i\right)^2\right] - \eps \\
\end{equation}
Next, observing that the collection of Gaussians $(g_{i,j})_{i,j}$ satisfy the premise of Theorem \ref{thm:geom}, we can apply Theorem \ref{thm:geom} to lower bound:
\begin{equation}		\label{eqn:e4}
	\frac{1}{N^2}\E\left[\left(\sum^n_{i = 1}\sum^{n_i}_{j = 1}g_{i,j}\right)^2\right]
	\geq \frac{2^{-Cd}}{N^2}\left(\sum^n_{i = 1}n_i\right)^2 
	= 2^{-Cd}\left(\sum^n_{i = 1}w'_i\right)^2,
\end{equation}
where in the last step, we again use the definition $n_i = w'_i N$. In the final step, we can use \eqref{eqn:e2} to further bound the above expression in terms of $\{w_i\}_i$:
\begin{equation}		\label{eqn:e5}
	2^{-Cd}\left(\sum^n_{i = 1}w'_i\right)^2 
	\geq 2^{-Cd} \left(\sum^n_{i = 1}w_i\right)^2 - \epsilon
	= 2^{-Cd}\cdot W^2 - \eps.
\end{equation}
Putting together the bounds from \eqref{eqn:e3}, \eqref{eqn:e4}, and \eqref{eqn:e5}, we get that:
\begin{align*}
	\E\left[\left(\sum^n_{i = 1}w_ig_i\right)^2\right] 
	&\geq \frac{1}{N^2}\E\left[\left(\sum^n_{i = 1}n_ig_i\right)^2\right] - \eps  \\
	&\geq 2^{-Cd}\left(\sum^n_{i = 1}w'_i\right)^2 - \eps \\
	&\geq 2^{-Cd} \cdot W^2 - 2\eps \\
	&\geq \frac{2^{-Cd}\cdot W^2}{2},
\end{align*}
where the last step follows from our choice of $\epsilon$.